\documentclass[twoside,a4paper]{article}

\usepackage{amsmath,graphicx,amssymb,fancyhdr,amsthm,enumerate,textcomp,amscd} 
\usepackage[usenames]{color}
\usepackage[all]{xy}
\newtheorem{thm}{Theorem}

\newtheorem{prop}[thm]{Proposition}
\newtheorem{con}[thm]{Conjecture} 
\theoremstyle{definition} 
\newtheorem{defn}[thm]{Definition}
  
\theoremstyle{remark}  
  
\def\beq{\begin{eqnarray}}  
\def\eeq{\end{eqnarray}}  
\def\bsp{\begin{split}}  
\def\esp{\end{split}}

\def\d{\mathrm{d}}

\begin{document}   
   
\title{\Large\textbf{On a new class of infinitesimal group actions on pseudo-Riemannian manifolds}}  
\author{{\large\textbf{Sigbj\o rn Hervik }    }
 \vspace{0.3cm} \\     
Faculty of Science and Technology,\\     
 University of Stavanger,\\  N-4036 Stavanger, Norway         
\vspace{0.3cm} \\      
\texttt{sigbjorn.hervik@uis.no} }     
\date{\today}     
\maketitle   
\pagestyle{fancy}   
\fancyhead{} 
\fancyhead[EC]{S. Hervik}   
\fancyhead[EL,OR]{\thepage}   
\fancyhead[OC]{Nil-Killing vectors}   
\fancyfoot{} 

\begin{abstract}
Using the Lie derivative of the metric we define a class of Lie algebras of vector fields by generalising the concept of Killing vectors. As a Lie algebra they define locally a group action on the pseudo-Riemannian manifold through exponentiation. The  motivation behind studying these infinitesimal group actions is the investigation of $\mathcal{I}$-degenerate pseudo-Riemannian spaces, i.e., spaces having identical polynomial curvature invariants. In particular, we show that all the known  examples of $\mathcal{I}$-degenerate pseudo-Riemannian spaces possess such vector fields.
\end{abstract}

The author of this Letter has long been puzzled by the following simple example in four neutral dimensions \cite{SHII,HHY}: 
 \[ ds^2=2\d u(\d v+V\d u)+2\d U(\d V+bv^4\d U),\]
 where $b$ is a constant. 
This example arises as a limit of a more general class of metrics and possesses only 3 Killing vectors and, hence, is not locally homogeneous. However, in spite of being inhomogeneous, all polynomial curvature invariants are constants (CSI) \cite{CSI}. CSI spaces have proven relevant for, for example, solutions of massive gravity theories \cite{CPS,CPS2,SS} and as spacetimes carrying gravitational waves \cite{KPZK, BM}. We are therefore led to the question:  Is there an underlying mathematical structure that ensures that this space is CSI?

In the Riemannian case, a space is CSI if and only if the space is locally homogeneous \cite{PTV}. The Lorentzian case, on the other hand, possesses CSI examples which do not have any Killing vectors and hence the CSI property cannot be explained using Killing vectors alone. The above example in four neutral dimensions is also such an example. 

In this Letter we will present a new concept and  mathematical structure which turns up to fully explain the CSI property of the above metric. This new concept has wider consequences as well, giving new insights into the structure of spaces having identical polynomial curvature invariants.  
 

Consider therefore the set, $\mathcal{I}$, of \emph{scalar polynomial curvature invariants} of a pseudo-Riemannian spacetime $(M,g)$. 
This set is finitely generated \cite{GW} using full contractions of the Riemann tensors and its covariant derivatives up to some finite order $k$. Thus, is sufficient to consider a finite set: 
\[ \mathcal{I}=\left\{ R, R_{abcd}R^{abcd},...,R_{abcd;e}R^{abcd;e},...\right\}\] 
Each polynomial invariant, $I\in\mathcal{I}$, is a smooth function, $I\in \Lambda^0(M)$, on the manifold; i.e., $I:M\rightarrow \mathbb{R}$, and, hence, the set of invariants $\mathcal{I}$, can be viewed as a function $\mathfrak{I}:M\rightarrow \mathbb{R}^N$, for some finite $N$. 

\begin{defn}
Given a pseudo-Riemannian space $(M,g)$. 
Then, an $\mathcal{I}$-preserving diffeomorphism (IPD), $\phi$, is a diffeomorphism $\phi:M\rightarrow M$ such that $\phi^*I=I$ for all $I\in \mathcal{I}$. 
\end{defn}
An IPD therefore compares the value of all invariants at different points since $\phi^*I=I\circ\phi$ as functions on $M$. Clearly, if the space possesses many such IPDs then the value of the invariants are identical for a large subset of $M$. It is obvious that such IPDs form a (pseudo-)group structure. 

We will particularly be interested in the case with there are one-parameter families, $\phi_t$, of IPDs so that $\phi_t^*I=I$ for $t\in\mathbb{R}$. Let $X$ be the vector field generating  $\phi_t$ so that $\phi_t=\exp(tX)$. Hence,
\[ \pounds_X I=\lim_{t\rightarrow 0} \frac 1t(\phi^* I-I)=0. \] 
One set of vector fields where this is satisfied for all $I$ is of course the set of Killing vector fields. The Lie derivative of the metric with respect to $X$ is zero iff $X$ is Killing: 
\[ \pounds _Xg_{ab}=0 \quad \Leftrightarrow \quad X \text{ is Killing.}\] 
In this case $\phi_t=\exp(tX)$ are (local) isometries. Clearly, there may be many IPDs that are not generated by Killing vectors, as an example, consider flat space for which all invariants are zero. Then any smooth $\phi$ would do since $\phi^*I$ is the zero-function.  On the other hand, if there is a transitive set of vectors so that $\mathrm{span}\{X_i\}|_p=T_pM$, where each $X_i$ generates an IPD, then the space is CSI in, at least, a neighbourhood of $p\in M$. Moreover, if the space is CSI then we would expect to find (locally) a transitive set of IPDs. 

It is therefore advantageous to generalise the Killing equation for isometries to get a local equation ensuring the existence of IPDs.  To this end we define \emph{Nil-Killing} vector fields as follows:
\begin{defn}
A vector field $X$ on $M$ for which $\pounds_{X}g_{ab}=n_{ab}$, where $n_{ab}$ as an operator ${\sf N}_X=(n^a_{~b}): TM\rightarrow TM$ is nilpotent, will be called a \emph{Nil-Killing field}. 
\end{defn}

In an index-free notation we define the operator ${\sf N}_X: TM\rightarrow TM$ by 
\[ g({\sf N}_X(Y),Z)=(\pounds_{X}g)(Y,Z), \quad \forall Y,Z \in\mathfrak{X}(M).\] 
The idea is that such vector fields are Killing vector fields as far as the polynomial invariants of $n_{ab}$ are concerned (i.e., all eigenvalues of ${\sf N}_X$ are zero). In particular, Killing implies Nil-Killing.

The set of Nil-Killing vectors does not necessarily form a (finite dimensional) Lie algebra. In subsequent papers we will study conditions when this happens. However, it is still useful to define a \emph{Nil-Killing Lie algebra}:  
\begin{defn}[Nil-Killing Lie algebra]
A {\it Nil-Killing Lie algebra} $\mathcal{N}$ is a set of smooth vector fields on $M$ fulfilling: 
\begin{enumerate}
\item{}  $\mathcal{N}$ forms a Lie algebra. 
\item{}  For all $X\in \mathcal{N}$, $X$ is Nil-Killing. 
\end{enumerate}
\end{defn}
The Lie algebra requirement implies that such vector fields generate (at least locally) a group action on the manifold through exponentiation. It is not clear under what conditions they generate IPDs, however interestingly, in \emph{almost} all examples\footnote{The only example where it does not necessarily generate an IPD is the Kundt case when it is not degenerate Kundt \cite{kundt}, see later.} given in this Letter, they do. 

As the approach here is a local approach, we will henceforth ignore all global issues. Therefore, the manifold $M$ is assumed to be  a sufficiently small neighbourhood, if necessary. 

Let us consider various examples in different signatures. In the positive-definite Riemannian case, Nil-Killing implies Killing. So all Riemannian examples are Killing vectors and hence, they generate (local) isometries. We will therefore consider first the Lorentzian case where non-trivial Nil-Killing vectors exist and seem to play an interesting r\^ole for a certain class of metrics. 
\par{Note:} In an early paper by Coll, Hildebrandt and Senovilla\cite{Sen1}, a class of symmetries preserving the Lorentzian Kerr-Schild metrics are defined which seems to be examples of  Nil-Killing vectors. Indeed, these were later used in \cite{Sen2} to establish the non-existence of trapping horizons in spacetimes with vanishing curvature invariants. 
\paragraph{Lorentzian case.}
The generic Lorentzian spacetime does not seem to allow for a Nil-Killing vector field, but there is an interesting class which does. 
\subparagraph{Kundt spacetimes} 
A Kundt spacetime is defined as a spacetime possessing a null-vector field $\ell$, which is non-expanding, non-shearing, and twistfree \cite{CSI,kundt}. This implies that we can write (using a null-frame $\{ \ell_a, n_a, m^{(i)}_{a}\}$, $i=1,...,n-2$): 
\[ \ell_{(a;b)}=L\ell_a\ell_b+\tau_i\ell_{(a}m^{(i)}_{b)}.\] 
Consequently, $\pounds_{\ell}g_{ab}=\ell_{(a;b)}$ is nilpotent as an operator. 
\begin{prop}
A Kundt spacetime possesses a null-vector field which is Nil-Killing. 
\end{prop}

\subparagraph{VSI spacetimes} 
Spacetimes for which all polynomial curvature invariants vanish are called  \emph{VSI spacetimes} \cite{VSI,Higher, CFHP}, and they can be written 
\beq\label{VSI}
ds^2=2\d u(\d v+H\d u+W_{i}\d x^i)+\delta_{ij}\d x^i\d x^j, \quad i=1,...,n-2,
\eeq
where 
\beq
H&=&\epsilon\frac{v^2}{2(x^1)^2}+vH^{(1)}(u,x^k)+H^{(0)}(u,x^k), \\
W_1&=& -\epsilon\frac{2v}{x^1}, \qquad
W_{i}= W^{(0)}_i(u,x^k), \quad i\neq 1,
\eeq
and $\epsilon=0,1$. 
\begin{thm}
Assume that an $n$-dimensional Lorentzian space has all vanishing curvature invariants (VSI). Then there exists a Nil-Killing Lie algebra $\mathcal{N}$ which is transitive; i.e.,  $\mathrm{dim}(\mathcal N|_p)=n$ for all $p\in M$. 
\end{thm}
\begin{proof}
There are two cases to consider, $\epsilon=0,1$. The set $\mathcal{N}$ can be given using the basis vectors (for both cases the set given is Abelian): 
\begin{description}
\item[$\epsilon=0$]: $\left\{ \partial_u,\partial_v,\partial_{x^1},\partial_{x^2},...,\partial_{x^{n-2}}\right\}$. 
\item[$\epsilon=1$]: $\left\{ \xi_1,\xi_2,\xi_3,\partial_{x^2},...,\partial_{x^{n-2}}\right\}$, 
\end{description} 
where 
\beq
\xi_1 &=& \frac{u(2(x^1)^2-uv)}{(x^1)^2}\partial_{x^1}-\frac {u^2}{x^1}\partial_u+\frac{(2(x^1)^4+2uv(x^1)^2-u^2v^2)}{(x^1)^3}\partial_v. \nonumber\\
\xi_2 &=& \frac{v}{(x^1)^2}\partial_{x^1}+\frac 1{x^1}\partial_u+\frac{v^2}{(x^1)^3}\partial_v. \nonumber \\
\xi_3 &=&\frac{(x^1)^2-uv}{(x^1)^2}\partial_{x^1}-\frac u{x^1}\partial_u+\frac{v((x^1)^2-uv)}{(x^1)^3}\partial_v. 
\eeq

Let us consider an arbitrary point $p$ given by $(u,v,x^k)=(u_0,v_0,x^k_0)$. Then clearly the $\epsilon=0$ case the vector fields all generate translations. By using the definition of the Lie derivative: 
\beq
\pounds_{\partial_{v}}g&=&2H^{(1)}\d u\d u,\\
 \pounds_{\partial_{\mu}}g&=&2\d u\left[(v\partial_{\mu}H^{(1)}+\partial_{\mu}H^{(0)})\d u+\partial_{\mu}W^{(0)}_j\d x^j\right],~~ \mu=u,x^i.
\eeq
The corresponding operator $n^a_{~b}$ is nilpotent (as can be seen since it has only zero-eigenvalues). 

For $\epsilon=1$ we need to use the $\xi_A$'s, $a=1,2,3$,  as well. Then, by computing the Lie derivatives, they are of the form: 
\[ \pounds_{\xi_A}g=2\d u\left[f_{(A)}(u,v,x^k)\d u+k_{(A)i}(u,v,x^k)\d x^i\right],
\]
which is again nilpotent as an operator. 
\end{proof}
In this case, both  Nil-Killing Lie algebras are Abelian. However, it should be noted that there exist larger Nil-Killing Lie algebras for these VSI spaces. These larger Nil-Killing Lie algebras are not Abelian, they can be solvable, semisimple, or even have non-trivial Levi-decomposition.

For spacetimes for which all polynomial curvature invariants are constants \cite{CSI}, all known examples seem to possess such a set. We therefore believe that: 
\begin{con}
\label{con:CSI}
Assume that an $n$-dimensional Lorentzian space has all constant curvature invariants (CSI). Then there exists a set $\mathcal{N}$ which is transitive; i.e., $\dim (\mathcal N|_p)=n$ for all $p\in M$. 
\end{con}

\paragraph{Other signatures.} 
Consider the following class in four dimensions of neutral signature \cite{HHY}: 
\beq
ds^2&=&2\d u\big[\d v+(aV+H(u,U))\d u\big]+2dU\big[dV+P(v,u,U)\d U\big], \nonumber\\
P(v,u,U) &=& bv^4+v^3G_3(u,U)+v^2G_2(u,U)+vG_1(u,U)+G_0(u,U),
\eeq
 where $a$ and $b$ are constants. Note that the example in the introduction is a member of this class. 

This class is CSI as well, but does not, in general, possess any Killing vectors. However, the following span a set $\mathcal{N}$: 
\[ \mathcal{N}=\mathrm{span}\{\partial_u,\partial_v,\partial_U,\partial_V\}. \] 
Computing the Lie derivatives: 
\beq
{\pounds_{\partial_i}g}&=&2\partial_iH(u,U)\d u^2+2\partial_iP(v,u,U)\d U^2, \quad  i=u,U,
\\ 
\pounds_{\partial_v}g&=& 2\left[4bv^3+3v^2G_3(u,U)+2vG_2(u,U)+G_1(u,U)\right]\d U^2, \\
\pounds_{\partial_V}g&=& 2a\d u^2.
\eeq
To see that all of these tensors are nilpotent as operators, we first define the one-parameter group of diffeomorphisms: 
\beq \label{diffboost}\phi_t: (u,v,U,V)\mapsto (e^{-t}u,e^{t}v,e^{-2t}U,e^{2t}V).\eeq 
We note that this map leaves the origin fixed and induces an element in $SO(2,2)$ with respect to the origin. When we evaluate the limit ($n_{ab}$ is any of the above Lie derivatives): 
\[ \lim_{t\rightarrow \infty} \phi_t^*n_{ab}=0, \] 
implying that all polynomial invariants of $n_{ab}$ are zero at the origin, hence the characteristic equation is trivial and $n_{ab}$ is nilpotent at the origin. If we consider a different point $p=(u_0,v_0,U_0,V_0)$, we perform a change of variables: 
\[ (\tilde{u},\tilde{v},\tilde{U},\tilde{V})=(u-u_0,v-v_0,U-U_0,V-V_0)\]
so that $p$ is at the origin of the tilded variables. However, the form of the Lie derivatives remains the same so that by using the transformation $\phi_t$ we get the conclusion that $n_{ab}$ is nilpotent at $p$ as well. Hence, $n_{ab}$ is everywhere nilpotent. 

These vectors (or rather, the diffeomorphisms they generate) can be used to show that this is indeed a set of CSI spacetimes since they all generate IPDs. 

We note that the set $\mathcal{N}$ is in this case abelian since all vector fields commute. However, there is a larger $\widehat{\mathcal{N}}\supset \mathcal{N}$, given by the additional vector field 
\[ \xi_5=-u\partial_u+v\partial_v-2U\partial_U+2V\partial_{V}.\]
 We note that $\xi_5$ generates the diffeomorphism $\phi_t$ given in eq.(\ref{diffboost}). Indeed, computing the Lie derivative: 
\[ \pounds_{\xi_5}g=-2\d u^2\left(u\partial_u+U\partial_U+2\right)H(u,U)-2\d U^2\sum_{n=0}^3v^n\left(4-n+u\partial_u+U\partial_U\right)G_n(u,U).\]
Hence, 
\[ \widehat{\mathcal{N}}=\mathrm{span}\{\partial_u,\partial_v,\partial_U,\partial_V,\xi_5\},\] 
is a (solvable) Lie algebra whose Lie derivatives of the metric $g$ are all nilpotent as operators. 

All known examples of CSI spaces seem to possess such a transitive set. Therefore, \emph{we believe that Conjecture \ref{con:CSI} is valid for arbitrary pseudo-Riemannian manifolds. }

\subparagraph{$\mathcal{I}$-degenerate pseudo-Riemannian manifolds. }
In \cite{HHY} we constructed a general class of $\mathcal{I}$-degenerate pseudo-Riemannian manifolds. This class contains all known examples of $\mathcal{I}$-degenerate pseudo-Riemannian manifolds; i.e. metrics allowing for a continuous metric deformation having identical polynomial curvature invariants. 

We will briefly review the metrics. Let $P(v_1,v_2,...,v_k)$ be a polynomial in the $v_i$'s with coefficients being arbitrary functions of $(u^i, x^a)$. Define $\mathcal{P}$ as  the ring of all such polynomials: 
\[ \mathcal{P}:=\left\{P(v_1,v_2,...,v_k) ~~| ~~ P ~ \text{polynomial,  coefficients depend on}~ (u^i, x^a)\right\} \]
We define subsets of this set and indicate them with a bracket $[-,..,-]$. The bracket consists of a list of monomials in $v_i$'s and indicates the highest allowable possible power of the $v_i$'s. For example, $[v_1^3,v_2v_3^5]\subset \mathcal{P}$ is the subset including the following powers: $v_1^n, n=0,...,3$, and $v_2^mv_3^q$, $m=0,1$ and $q=0,...,5$. 

The metrics can now be written 
\beq\label{result}
g=2du^i\left(a_{ij}dv^j+A_{ij}du^j+B_{ia}dx^a\right)+g_{ab}dx^adx^b,
\eeq
where $g_{ab}=g_{ab}(u^i,x^a)$ and $a_{ij}$,  $A_{ij}$, and $B_{ia}$, are polynomials belonging to some finite  subset $[-,...,-]$ of $\mathcal{P}$, with  arbitrary smooth coefficients in $(u^i,x^a)$. See \cite{HHY} for the construction and the conditions these polynomial functions obey. 

It is now straight-forward to see that the following set is a Nil-Killing Lie algebra: 
\[ \mathcal{N}=\mathrm{span}\{\partial_{v_1},\partial_{v_2},...,\partial_{v_k}\}.\] 
Clearly, each $\partial_{v_i}$ generates an IPD: 
\[ \phi_\lambda=\exp(\lambda \partial_{v_i}),\] 
which is simply a translation $v_i\mapsto v_i+\lambda$. This set can thus be of dimension equal to the real rank of the pseudo-orthogonal group $O(p,q)$ (the real rank is equal to ${\rm min}(p,q)$.) 

\paragraph{Open questions. } Many questions remain, some of them are:
\begin{enumerate}
\item 	{Under what conditions does the set of Nil-Killing vectors form a Lie algebra?} 
\item Does a generic space allow for a Nil-Killing vector?  
\item Does the existence of a non-Killing Nil-Killing vector imply an $\mathcal{I}$-degenerate space? 
\item Given two $\mathcal{I}$-degenerate metrics with identical invariants, do they have identical maximal Lie algebras $\mathcal{N}$?    
\item Under what conditions do Nil-Killing vector fields generate IPDs? 
\end{enumerate}
In future work we will investigate some of these questions.

\paragraph{Acknowledgments.} 
This work was supported from the Research Council of Norway, Toppforsk
grant no. 250367: \emph{Pseudo-Riemannian Geometry and Polynomial Curvature Invariants:
Classification, Characterisation and Applications}.


\end{document}